\documentclass[prodmode]{acmsmall}
\usepackage{fontenc}
\usepackage[utf8]{inputenc}

\usepackage{url}
\usepackage{algorithm2e}
\usepackage{subfigure}

\usepackage[numbers]{natbib}
\usepackage{mathtools}
\usepackage{amsmath}
\usepackage{amsfonts,graphpap,amscd,mathrsfs,graphicx,lscape}
\usepackage{epsfig,amssymb,amstext,xspace}
\usepackage[dvipsnames,usenames]{xcolor}
\usepackage{prettyref}
\usepackage{graphicx} % support the \includegraphics command and options

\usepackage{booktabs} % for much better looking tables
\usepackage{array} % for better arrays (eg matrices) in maths
\usepackage{paralist} % very flexible & customisable lists (eg. enumerate/itemize, etc.)
\usepackage{verbatim} % adds environment for commenting out blocks of text & for better verbatim
\newtheorem{observation}{Observation}

\newcommand{\cout}[1]{}
\newcommand{\poly}{\text{poly}}

\newcommand{\mattnote}[1]{{#1}}

\title{Revenue Maximization and Ex-Post Budget Constraints}
\author{CONSTANTINOS DASKALAKIS$^1$
	\affil{M.I.T}
	NIKHIL R. DEVANUR$^2$
	\affil{Microsoft Research}
	S. MATTHEW WEINBERG$^3$
	\affil{Princeton University}
}

\begin{document}

\begin{abstract}

We consider the problem of a revenue-maximizing seller with $m$ items for sale to $n$ additive bidders with hard budget constraints, assuming that the seller has some prior distribution over bidder values and budgets. The prior may be correlated across items and budgets of the same bidder, but is assumed independent across bidders. We target mechanisms that are Bayesian Incentive Compatible, but that are \emph{ex-post} Individually Rational and \emph{ex-post} budget respecting. Virtually no such mechanisms are known that satisfy all these conditions and guarantee any revenue approximation, even with just a single item. We provide a computationally efficient mechanism that is a $3$-approximation with respect to all BIC, ex-post IR, and ex-post budget respecting mechanisms. Note that the problem is NP-hard to approximate better than a factor of $16/15$, even in the case where the prior is a point mass~\cite{ChakrabartyGoel}. We further characterize the optimal mechanism in this setting, showing that it can be interpreted as a \emph{distribution over virtual welfare maximizers.}

We prove our results by making use of a black-box reduction from mechanism to algorithm design developed by~\cite{CaiDW13b}. Our main technical contribution is a computationally efficient $3$-approximation algorithm for the algorithmic problem that results by an application of their framework to this problem. The algorithmic problem has a mixed-sign objective and is NP-hard to optimize exactly, so it is surprising that a computationally efficient approximation is possible at all. In the case of a single item ($m=1$), the algorithmic problem can be solved exactly via exhaustive search, leading to a computationally efficient exact algorithm and a stronger characterization of the optimal mechanism as a distribution over virtual value maximizers.

\end{abstract}

\begin{bottomstuff}
$^1$ costis@csail.mit.edu, supported by a Microsoft Research Faculty Fellowship and NSF Awards CCF-0953960 (CAREER) and CCF-1101491.\\
$^2$ nikdev@microsoft.com.\\
$^3$ sethmw@cs.princeton.edu.
\end{bottomstuff}
\maketitle
\section{Introduction}

Most of auction theory crucially depends on the assumption of quasi-linear utilities, that the utility is equal to  valuation minus payments. 
This assumption fails when bidders are budget constrained.\footnote{The terms financially constrained  bidders or bidders with liquidity constraints are used synonymously.} 
%Budget constraints on bidders is a crucial aspect of many auctions, 
Auctions with budget constrained bidders are  commonplace, and prominent examples of this are ad-auctions and auctions for government licensing such as the FCC spectrum auction. 
An interesting example of budget constraint occurs in the auction for professional cricket players in the Indian Premier League: the league imposes a budget constraint on all the teams as a means of ensuring well balanced teams. 
Another source of budget constraints is what \citet{CheGale1998} call the {\em moral hazard problem}: procurement is often delegated and budget constraints are imposed as a means of controlling the spend. 
%In ad auctions, budget constraints arise as a consequence of scale, as the number of potential winnable items (ad clicks) could be very large. 
A budget represents the bidder's \emph{ability to pay}, in contrast to the valuation which represents his {\em willingness to pay}. 
For this reason, budgets may be more tangible and easier to estimate than valuations. 
It is therefore important to understand how budget constraints impact the design of auctions; this has been well established by now \citep{CheGale1998, PaiVohra, BenoitKrishna, LaffontRobert, Maskin, MalakhovVohra, CheGale2000, BhattacharyaGGM10}. 

The theory of auctions in the presence of budget constraints on bidders lags far behind the theory of auctions without budgets.
For instance, consider the design of {\em optimal}  (revenue maximizing) auctions that are Bayesian incentive compatible (BIC) and ex-post individually rational (IR). 
While \citet{Myerson} gives a beautiful theory characterizing the optimal auction for any single parameter domain, no such characterization is known in the presence of private budgets (that could be correlated with the valuation). 
As a way to deal with this difficulty, previous papers have considered special cases and auctions with a subset of the desired properties. (See \prettyref{sec:related} for details.) We adopt the Computer Science approach of approximation, while incorporating all the desired properties. {\em The main result of this paper is a 3-approximation to the optimal auction in the class of auctions that are 
	\begin{itemize}
		\item BIC,
		\item ex-post IR and 
		\item ex-post budget respecting, with private budgets that could be correlated with the valuations, 
	\end{itemize} 
for multiple heterogenous items and additive valuations.} This is the first constant factor approximation for this class of auctions. 
Moreover, the computational problem, even without any incentive constraints is already NP-Hard to approximate within a ratio of $16/15$ \citep{ChakrabartyGoel}. This too suggests that an approximation is necessary. 

\subsection{Overview of Techniques}
{We prove our main result by making use of an algorithmic framework developed in~\cite{CaiDW13b}. The computational aspect of their framework provides a \emph{black-box reduction} from a wide class of Bayesian mechanism design problems to problems of purely algorithm design. More specifically, they show that any $\alpha$-approximation algorithm for a certain incentive-free algorithmic problem (induced by the mechanism design problem at hand) can be leveraged to find a BIC,  IR mechanism that is also an $\alpha$-approximation (to the optimal BIC, IR mechanism) in polynomial time. Significant further details on their reduction and how to employ it can be found in Section~\ref{sec:CDW}. After applying their framework to our problem, there is still the issue of solving the algorithmic problem that pops out of the reduction. This turns out to be essentially a (virtual) welfare maximization problem (without budgets), but where bidder types are somewhat involved. The optimization involves a mixed sign objective (i.e. the objective is a sum of several terms which can be positive or negative). Such optimization problems are typically solvable \emph{exactly} in polynomial time or computationally hard to approximate within any finite factor, but rarely in between (due to the mixed signs in the objective). Interestingly, we obtain a 3-approximation for our mixed-sign objective problem despite the fact that it is NP-hard to optimize exactly. The design and analysis of our algorithm can be found in Section~\ref{sec:main}.}

{Cai et. al.'s framework also contains a structural result. We use it to show that the optimal auction in our setting is a \emph{distribution} over virtual welfare maximizers. By this, we mean that the optimal mechanism maintains a distribution over $n$ mappings, one mapping per bidder that maps types to virtual types, and, given a vector of reported types, it samples $n$ mappings from this distribution, uses them to map the reported types to virtual types, and proceeds to choose an allocation that optimizes virtual welfare. Note that by virtual types in the previous sentence we do \emph{not} mean the specific virtual types as computed by Myerson's virtual transformation, which aren't even defined for multi-dimensional types, but just \emph{some} virtual types that may or may not be the same as the true types. In particular, each mapping in the support of the mechanism's distribution will take as input a type (which is an additive function with non-negative item values plus a non-negative budget), and output a virtual type without a budget constraint and whose valuation function is the sum of a budgeted-additive function\footnote{{A function $v(\cdot)$ is budgeted-additive if there exists a $b$ such that $v(S) = \min\{b,\sum_{i \in S} v(\{i\})\}$ for all $S$. Note that this is different from an additive buyer with a budget, and that a budgeted-additive buyer indeed has quasi-linear utilities.}} with non-negative item values (which depends on the input type in a very structured way) plus an additive function with possibly negative item values (which may be unstructured with respect to the input type). We provide a formal statement of this structural claim in Section~\ref{sec:main} as well.}
Note that for the special case of a single item auction, this gives a particularly simple structure: the virtual types are now {just a single (possibly negative) real number, which could be interpreted as a virtual value.} The optimal auction simply maps reported types to virtual {\em values} and assigns the item to the bidder with the highest virtual value.

\subsection{Related Work}
\label{sec:related} 

The result that comes closest to characterizing the optimal auction is that of \citet{PaiVohra}: they characterize the optimal budget respecting BIC auction for a single item. Their auction is implemented as an all-pay auction and is therefore not ex-post IR.
%, and they assume that the budgets and values are \emph{independently} distributed.\footnote{{It seems from their comments that they relax this assumption but it is very technical to do so.}}
 {They show that the optimal BIC, \emph{interim} IR mechanism that respects budgets ex-post takes on a form similar to Myerson's, but with additional \emph{pooling} to enforce that no bidder is asked to pay more than her budget, while also maintaining that no bidder has incentive to underreport their budget.} Earlier, \citet{LaffontRobert} and \citet{Maskin} considered the case where valuations are private information but budgets are common knowledge and identical. 
\citet{MalakhovVohra} study the setting where there are two bidders, one has a known budget constraint while the other does not.
\citet{CheGale2000} characterize the optimal pricing scheme for a single item with a single bidder,  with private valuation and budget that could be correlated with each other. 
The limited special cases considered by these papers point to the difficulty of characterizing the optimal auction, which motivates the search for efficient approximations. 

Another line of work ranks different auction formats by the revenue generated in the presence of budgets.  
\citet{CheGale1998} compare first price, second price and all-pay auctions, while 
\citet{BenoitKrishna} compare sequential and simultaneous auctions.%Auctions using reduced form/interim allocations. 

In the computer science tradition, \cite{BhattacharyaGGM10} give a 4-approximation for multiple items with additive valuations, but they assume that the budgets are publicly known, and the auction is not ex-post IR. 
\citet{ChawlaMM11} give a 2-approximation in a single parameter domain, but assume that the budgets are public. 
%- Multiparameter matroid, public budget O(1) 
They also consider private budgets, where budgets and values are independently distributed, in single parameter matroid domains, and MHR Distributions, and give a $3(1+e)$-approximation. Finally,~\citet{CaiDW12}, provide exactly optimal mechanisms for multiple items, additive valuations and private budgets, but their auctions are interim-IR.
Once again, all these auctions make additional assumptions when compared to us. 

\citet{CaiDW13b} give a general reduction from mechanism design to algorithm design, which we use for our results. 
For the special case of a single item auction with private budgets, we show that the algorithmic problem obtained through this reduction is quite easy to solve optimally, resulting in exactly optimal single-item auctions with budgets. 
However, when there are multiple items the resulting algorithmic problem becomes NP-Hard \citep{ChakrabartyGoel}. We give a 3-approximation to this algorithmic problem which through the reduction gives a 3-approximately optimal multi-item auction with budgets.
Recently, \cite{Bhalgatetal} showed that (a weaker form of) the reduction of \citet{CaiDW13b} could be obtained using the simpler multiplicative weight update method instead of the ellipsoid algorithm used originally, and consider the variant of our setting where the items are \emph{divisible}. The algorithmic problem in this case is once again easy. \citet{DaskalakisW15} also use the reduction in \cite{CaiDW13b} to design an auction for a non-linear objective, namely the \emph{makespan} of an assigment of jobs to machines. %\nikhil{Check this for mistakes.}

The auction design problem has also been considered in a \emph{worst-case} model, as opposed to a Bayesian model. 
 A standard framework is that of {\em competitive} auctions, where a bound is shown on the ratio of the revenue of an optimal auction to the revenue of the given auction on any instantiation of valuations and budgets.  
 \citet{Borgsetal2005} and \citet{Abrams2006} give constant competitive auctions for multi-unit auctions, under an assumption of \emph{bidder dominance}, that the contribution of a single bidder to the total revenue is sufficiently small. 
\cite{DevanurHH13} give constant competitive auctions for single parameter downward-closed domains with a public, common budget constraint. 
Since the worst-case setting is decidedly more difficult than the Bayesian setting, these results are not comparable to ours. 
Another line of work considers the design of \emph{Pareto-optimal} auctions:  \citet{DobzinskiLaviNisan} characterize single item auctions that are Pareto-optimal, with public budgets and show an impossibility of a similar auction for private budgets. 
\citet{GoelMirrokniPaes-Leme} extend this auction to a more general poly-matroidal setting.

\subsection{Conclusions and Future Work}
{The goals of revenue-optimality, ex-post individual rationality, and ex-post budget feasibility seem to be at odds with one another. This is highlighted by the fact that, prior to our work, no known auctions even approximately satisfied all three conditions, even with just a single item and private budgets that are independent of values. We provide a computationally efficient 3-approximation for the significantly more general case of auctions for multiple heterogeneous goods and additive bidders with private budgets that can be correlated with their values. While this model is already quite general compared to the previous state-of-the-art, it is an important direction to see if our results can be extended to more complex classes of bidder valuations, or to more complex constraints on feasible allocations.}
In particular, well studied classes of valuations such as gross substitutes would be interesting next steps.

 %\nikhil{A virtual type has 4 types of parameters, the valuations and a budget which are non-negative, and a set of virtual valuations and a price multiplier which could be negative.    	
%	The optimal mechanism, on every profile, samples a mapping from this distribution, and then allocates the items in a way that maximizes virtual welfare 
%	plus virtual revenue.  Virtual welfare is the sum of the virtual value terms, and virtual revenue is the sum of the value terms subject to the budget constraint and  multiplied by the price %multiplier.}

%!tex root = main.tex
\section{Preliminaries}
 We begin with formal definitions of the mechanism design problem we study. 
 We then outline the reduction of \citet{CaiDW13b} (\prettyref{sec:CDW}) 
 and its implications (\prettyref{sec:instantiation}) for our problem. 
Finally we state a related problem (\prettyref{sec:gap}), the Generalized Assignment Problem, which we use in the design of our algorithm. 
\paragraph{Bidders} There are $n$ bidders, each with additive valuations over $m$ items and a hard budget constraint. Specifically, bidder $i$ has value $v_{ij}$ for item $j$, value $\sum_{j\in S} v_{ij}$ for set $S$, and hard budget $b_i$. We denote by $\vec{v}_i$ the vector of bidder $i$'s values for all $m$ items. We denote by $\mathcal{D}_i$ the joint distribution of $(\vec{v}_i, b_i)$. We denote by $\mathcal{D} = \times_i \mathcal{D}_i$ the joint distribution of all bidders' valuations and budgets. 

\paragraph{Mechanisms} Our goal is to design \emph{Bayesian Incentive Compatible} (BIC) mechanisms that are \emph{ex-post Individually Rational} (IR) and that respect budgets ex-post. Formally, for a (randomized) mechanism $M$, we can denote by $x^M_{ij}(\vec{v},\vec{b},r)$ to be $1$ if bidder $i$ receives item $j$ when the profile of values/budgets reported to $M$ is $(\vec{v},\vec{b})$ and the random seed used by $M$ is $r$, or $0$ otherwise. Similarly, we denote by $q^M_{i}(\vec{v},\vec{b},r)$ to be the price paid by bidder $i$ in the same conditions. We can then define the interim allocation probability $\pi^M_{ij}(\vec{v}_i,b_i)$ to be the probability that bidder $i$ receives item $j$ when reporting $(\vec{v}_i,b_i)$ over the randomness of other agent's (valuation,budget)s $(\vec{v}_{-i},\vec{b}_{-i})$ being drawn from $\mathcal{D}_{-i}$, and any randomness in $M$. We can similarly define the interim price $p^M_i(\vec{v}_i,b_i)$ to be the expected payment made by bidder $i$ over the same randomness. Formally, $\pi^M_{ij}(\vec{v}_i,b_i) = \mathbb{E}_{(\vec{v}_{-i},\vec{b}_{-i})\leftarrow \mathcal{D}_{-i}, r}[x^M_{ij}(\vec{v_i};\vec{v}_{-i},b_i;\vec{b}_{-i},r)]$ and $p^M_{ij}(\vec{v}_i,b_i) = \mathbb{E}_{(\vec{v}_{-i},\vec{b}_{-i})\leftarrow \mathcal{D}_{-i}, r}[q^M_{ij}(\vec{v_i};\vec{v}_{-i},b_i;\vec{b}_{-i},r)]$. Formal definitions of BIC, IR, and ex-post budgets are below.

\begin{definition}\label{def:BIC}(Bayesian Incentive Compatible) A mechanism $M$ is BIC if for all bidders $i$, and types $(\vec{v}_i,b_i), (\vec{v}'_i, b'_i)$ the following holds:
$$\vec{v}_i \cdot \vec{\pi}^M_i(\vec{v}_i, b_i) - p^M_i(\vec{v}_i, b_i) \geq \vec{v}_i \cdot \vec{\pi}^M_i(\vec{v}'_i, b'_i) - p^M_i(\vec{v}'_i,b'_i).$$

\noindent A mechanism is said to be $\epsilon$-BIC if for all bidders $i$, and types $(\vec{v}_i,b_i), (\vec{v}'_i, b'_i)$ the following holds:
$$\vec{v}_i \cdot \vec{\pi}^M_i(\vec{v}_i, b_i) - p^M_i(\vec{v}_i, b_i) \geq \vec{v}_i \cdot \vec{\pi}^M_i(\vec{v}'_i, b'_i) - p^M_i(\vec{v}'_i,b'_i) - \epsilon.$$
\end{definition}

\begin{definition}\label{def:IIR}(Interim/Ex-Post Individually Rational) A mechanism $M$ is interim IR if for all bidders $i$, and types $(\vec{v}_i,b_i)$ the following holds:
$$\vec{v}_i \cdot \vec{\pi}^M_i(\vec{v}_i, b_i) \geq p^M_i(\vec{v}_i, b_i) .$$
Further, it is ex-post IR if for all bidders $i$, all profiles $(\vec{v},\vec{b})$ and random seeds $r$, we have:
$$  \vec{v}_i \cdot \vec{x}^M_i(\vec{v},\vec{b},r)\geq q^M_i(\vec{v},\vec{b},r).$$
\end{definition}

\begin{definition}(Ex-Post Budget Respecting) A mechanism $M$ respects budgets ex-post if for all type profiles $(\vec{v},\vec{b})$, all random seeds $r$, and all bidders $i$ we have:
$$q^M_i(\vec{v},\vec{b},r) \leq b_i.$$
\end{definition}

\begin{definition}(No Positive Transfers) A mechanism $M$ has no positive transfers if for all type profiles $(\vec{v},\vec{b})$, all random seeds $r$, and all bidders $i$ we have:
$$q^M_i(\vec{v},\vec{b},r) \geq 0.$$

\end{definition}

\subsection{Reduction from Mechanism to Algorithm Design}\label{sec:CDW} In recent work,~\cite{CaiDW13b} provide an algorithmic framework for mechanism design, showing how to design mechanisms by solving purely algorithmic problems. We use this reduction to reduce our mechanism design problem to an algorithm design problem and show a 3-approximation to this algorithmic problem. In the rest of this section, we state the  general formulations of the mechanism design  and the corresponding algorithm design problems considered by~\citet{CaiDW13b}. Then we give the precise statement of their reduction, and a structural characterization of the optimal mechanism obtained as a byproduct of their reduction. Finally we instantiate these to state the corresponding problems in our setting, and massage the resulting problems to simplify them.

\cite{CaiDW13b} call the mechanism design problems of study BMeD($\mathcal{F},\mathcal{V},\mathcal{O}$),\footnote{BMeD stands for \textbf{B}ayesian \textbf{Me}chanism \textbf{D}esign.} where feasibility constraints $\mathcal{F}$, possible valuations $\mathcal{V}$, and optimization objective $\mathcal{O}$ parameterize the problem. Formally, this problem is defined as:\\

\noindent\textbf{BMeD($\mathcal{F},\mathcal{V},\mathcal{O}$)}:\\
\textsc{Input}: For each bidder $i \in [n]$, a finite set $T_i \subseteq \mathcal{V}$, and a distribution $D_i$ over $T_i$, presented by explicitly listing all types in $T_i$ and their corresponding probability.\\
\textsc{Output}: A feasible (selects an allocation in $\mathcal{F}$ with probability $1$), BIC, (interim) IR mechanism for bidders drawn from $D = \times_i D_i$.\\
\textsc{Goal}: Find the mechanism that optimizes $\mathcal{O}$ in expectation, with respect to all BIC, IR mechanisms (when bidders with types drawn from $D$ play truthfully).\\
\textsc{Approximation}: An algorithm is said to be an $(\epsilon,\alpha)$-approximation if it finds an $\epsilon$-BIC mechanism whose expected value of $\mathcal{O}$ (when bidders drawn from $D$ report truthfully) is at least $\alpha\cdot \text{OPT} - \epsilon$.\\

In our problem, the feasible allocations are those that award each item to at most one bidder. So we could denote the set of feasible allocations as $[m+1]^n$ (with the convention that selecting the allocation $\vec{a}$ awards item $j$ to bidder $a_j$ if $a_j > 0$, or no one if $a_j = 0$). The possible bidder types are all additive functions over items (with non-negative multipliers), and non-negative budgets, which we could denote by $\mathbb{R}^{m+1}_+$. Our objective is revenue. To ensure that all feasible mechanisms are ex-post IR (note that their reduction only guarantees interim IR without extra work) and ex-post budget respecting, we will define the objective function $\textsc{Revenue}$ as follows. $\textsc{Revenue}$ takes as input a valuation profile $(\vec{v},\vec{b})$, an allocation $\vec{x}$ (where $x_{ij} = 1$ iff bidder $i$ is awarded item $j$), and a price vector $\vec{p}$. We define $\textsc{Revenue}(\vec{v},\vec{b},\vec{x},\vec{p}) = \sum_i p_i$, if $0 \leq p_i \leq \min\{b_i, \vec{v}_i\cdot \vec{x}_i\}$ for all $i$, or $\textsc{Revenue}(\vec{v},\vec{b},\vec{x},\vec{p}) = -\infty$ otherwise. 

{There is a subtle issue with respect to \emph{why} we want to design mechanisms that respect budgets ex-post. Specifically, is it just because the designer wishes to offer this guarantee to the bidders, who have true quasi-linear preferences? If so, then this is exactly the setting we have described so far: the designer is constrained to select a mechanism that respects budgets ex-post, but bidders will still choose how to play as if they were quasi-linear. While this motivation is certainly mathematically interesting, it is also non-standard and perhaps unrealistic. Instead, the more common motivation is because bidders physically can't pay more than their budget, and would have utility $-\infty$ if asked to do so. In this case, the designer should actively \emph{exploit} this to extract higher revenue. For example, if for all $i$, bidder $i$ is awarded all the items and charged her budget with tiny probability $\epsilon /n$, then the designer needn't worry about bidders overreporting their budget (as otherwise they'd get utility $-\infty$ with probability $\epsilon /n$). Therefore, a mechanism can be made BIC in the latter case (while losing arbitrarily little revenue) iff for all $(\vec{v}_i, b_i), (\vec{v}'_i, b'_i)$ \textbf{with $b'_i \leq b_i$}, $\vec{v}_i \cdot \vec{\pi}^M_i(\vec{v}_i, b_i) - p^M_i(\vec{v}_i, b_i) \geq \vec{v}_i \cdot \vec{\pi}^M_i(\vec{v}'_i, b'_i) - p^M_i(\vec{v}'_i,b'_i).$ Note that this is a relaxed condition of the former setting, which requires the inequality to hold \textbf{for all $b_i, b'_i$}. Fortunately, the Cai et. al. framework applies in both settings, and the resulting structure and algorithmic problem are exactly the same. So all of our theorems, exactly as stated, hold in both of the described settings.}

Informally, the main result of~\cite{CaiDW13b} states that, for all $\mathcal{F}, \mathcal{V}, \mathcal{O}$, the problem BMeD($\mathcal{F},\mathcal{V},\mathcal{O}$) can be solved in polynomial time with black-box access to a poly-time algorithm for a purely algorithmic problem that they call GOOP($\mathcal{F},\mathcal{V},\mathcal{O}$).\footnote{GOOP stands for \textbf{G}eneralized \textbf{O}bjective \textbf{O}ptimization \textbf{P}roblem.} Below, $\mathcal{V}^\times$ denotes the closure of $\mathcal{V}$ under addition and (possibly negative) scalar multiplications (so for instance, $(\mathbb{R}^m_+)^\times = \mathbb{R}^m$).\\

\noindent\textbf{GOOP($\mathcal{F},\mathcal{V},\mathcal{O}$)}:\\
\textsc{Input}: A type $t_i \in \mathcal{V}$, multiplier $m_i \in \mathbb{R}$, {and virtual valuation function $g_i \in \mathcal{V}^\times$ for each $i \in [n]$.}\footnote{For other applications, the inputs $g_i(\cdot)$ are sometimes called instead cost functions.} \\
\textsc{Output}: An allocation $x \in \mathcal{F}$ and price vector $\vec{p} \in \mathbb{R}_+^n$.\\
\textsc{Goal}: Find $\arg \max_{x \in \mathcal{F},\vec{p}}\{\mathcal{O}(\vec{t},x,\vec{p}) + \sum_i m_i p_i + \sum_i g_i(x)\}$.\\
\textsc{Approximation}: $(x^*,\vec{p}^*)$ is said to be an $\alpha$-approximation if $\mathcal{O}(\vec{t},x^*,\vec{p}^*) + \sum_i m_i p_i^*+ \sum_i g_i(x^*) \geq \alpha \cdot \arg \max_{x\in \mathcal{F},\vec{p}}\{\mathcal{O}(\vec{t},x,\vec{p}) + \sum_i m_i p_i+ \sum_i g_i(x)\}$.\\

Further below we provide much more detail on the structure of the algorithmic focus of this paper, GOOP($[n+1]^m,\mathbb{R}^{m+1}_+,\textsc{Revenue}$), but we first conclude our discussion of the reduction we employ. The main result of~\cite{CaiDW13b} states that for all $\epsilon > 0$, an $(\epsilon,\alpha)$-approximation for BMeD($\mathcal{F},\mathcal{V},\mathcal{O}$) can be obtained from a poly-time $\alpha$-approximation for GOOP($\mathcal{F},\mathcal{V},\mathcal{O}$). The additive error (and failure probability in the theorem statement) is due to a sampling procedure in the execution of the reduction. We provide a full statement of their main result below.\footnote{The theorem statement is identical in content, but reworded for clarity and cleanliness.}

\begin{theorem}\label{thm:CDW}(Theorem~4 of~\cite{CaiDW13b}) For all $\mathcal{F},\mathcal{V},\mathcal{O}$, and $\epsilon > 0$, if there is a poly-time $\alpha$-approximation algorithm, $G$, for GOOP($\mathcal{F},\mathcal{V},\mathcal{O}$), there is a poly-time $(\epsilon,\alpha)$-approximation algorithm for BMeD($\mathcal{F},\mathcal{V},\mathcal{O}$) as well. Specifically, if $\ell$ denotes the input length of a BMeD($\mathcal{F},\mathcal{V},\mathcal{O}$) instance, the algorithm runs in time $\poly(\ell, 1/\epsilon)$, makes $\poly(\ell, 1/\epsilon)$ black box calls to $G$ on inputs of size $\poly(\ell, 1/\epsilon)$, and succeeds with probability $1-\text{exp}(-\poly(\ell, 1/\epsilon))$.
\end{theorem}

\cite{CaiDW13b} prove Theorem~\ref{thm:CDW} above by considering a linear program that optimizes over the space of interim forms that are both truthful (that satisfy the linear constraints in Definitions~\ref{def:BIC} and~\ref{def:IIR}), and feasible (those that correspond to an actual mechanism that selects an outcome $x \in \mathcal{F}$ on every profile with probability $1$).\footnote{In fact, they need to work with a generalization of interim forms, called {\em implicit forms}, to accommodate non-additive valuations. But we describe their proof for additive valuations for clarity of exposition, and because it is relevant for our setting.} Linear constraints enforcing that an interim form is BIC and interim IR can be written explicitly, but a computationally efficient separation oracle for the space of feasible interim forms is still required in order to solve the linear program. They show how to obtain such a separation oracle with black-box access to an algorithm that solves GOOP, and that this entire process preserves approximation as well.

\cite{CaiDW13b} further provide a structural characterization of the space of all feasible mechanisms (truthful or not), leading to a structured implementation of whatever interim form is output by the LP. Specifically, they show that the extreme points of the space of feasible interim forms correspond to mechanisms that associate a virtual valuation function $g_i(t_i)(\cdot)$ and price multiplier $m_i(t_i)$ to each type $t_i \in T_i$, and then selects on profile $(t_1,\ldots,t_n)$ the allocation and price vector that solves GOOP on input $t_1,\ldots,t_n$, $m_1(t_1),\ldots,m_n(t_n)$, $\sum_i g_i(t_i)(\cdot)$. They show further that solving the linear program explicitly finds a list of virtual valuation functions and multipliers whose resulting interim forms contain the optimal (truthful) interim form in their convex hull. Theorem~\ref{thm:CDWstructure} below captures the structural aspect of their result.

%\mattnote{Did we explicitly state this somewhere besides our Encyclopedia article?}
\begin{theorem}\label{thm:CDWstructure}(Implicit in~\cite{CaiDW13b}) For all BMeD instances, the optimal mechanism can be implemented as a distribution over generalized objective optimizers. Specifically, there exists a distribution $\Delta$ over mappings $(f^\delta_1,\ldots,f^\delta_n)$. Each mapping $f^\delta_i$ takes types $t_i$ in $T_i$ to price multipliers $m^\delta_i(t_i) \in \mathbb{R}$ and virtual valuation functions $g^\delta_i(t_i)(\cdot)  \in \mathcal{V}^\times$. The optimal mechanism first samples $(f^\delta_1,\ldots,f^\delta_n)$ from $\Delta$, and on profile $\vec{t}$, selects the outcome and price vector $\arg\max_{x \in \mathcal{F},\vec{p}}\{\mathcal{O}(\vec{t},x,\vec{p})+\sum_i m^\delta_i(t_i)\cdot p_i + \sum_i g^\delta_i(t_i)(x)\}$. 
\end{theorem}

In the section below, we provide further details surrounding instantiations of Theorems~\ref{thm:CDW} and~\ref{thm:CDWstructure} as they pertain to the problem at hand.

\subsection{Instantiations}\label{sec:instantiation} 
The goal of this section is to provide more details of the instantiation of Theorems~\ref{thm:CDW} and~\ref{thm:CDWstructure} to our setting, but not to provide proofs (for which we refer the reader to~\cite{CaiDW13b}). We begin by describing the linear program that the reduction of~\cite{CaiDW13b} would try to solve for our setting. Below, $F([n+1]^m,\mathbb{R}^{m+1}_+,\textsc{Revenue})$ denotes the space of interim forms of all feasible  (not necessarily truthful) mechanisms. Specifically, $(O,\vec{\pi},\vec{p}) \in F([n+1]^m,\mathbb{R}^{m+1}_+,\textsc{Revenue})$ if and only if there is a mechanism $M$ that awards each item at most once on every profile, is ex-post IR and ex-post budget respecting, awards bidder $i$ item $j$ when she reports type $t_i$ with probability exactly $\pi_{ij}(t_i)$ (w.r.t. all other bidders' types and the randomness in the mechanism) and chargers bidder $i$ price $p_i(t_i)$ in expectation (over all other bidders' types and the randomness in the mechanism), and whose expected revenue is exactly $O$. With this definition in mind, the linear program they solve is stated below.\\

\noindent\textbf{Variables:}
\begin{itemize}
\item $O$, denoting the expected revenue of the interim form found.
\item $\pi_{ij}(t_i)$ for all bidders $i$, items $j$, types $t_i$, denoting the probability that bidder $i$ receives item $j$ when reporting type $t_i$.
\item $p_i(t_i)$ for all bidders $i$ and types $t_i$, denoting the expected price paid by bidder $i$ when reporting type $t_i$.
\end{itemize}
\textbf{Constraints:}
\begin{enumerate}
\item $\sum_j \pi_{ij}(t_i) \cdot v_{ij}(t_i) - p_i (t_i) \geq \sum_j \pi_{ij}(t'_i) \cdot v_{ij}(t_i) - p_i(t'_i)$, for all bidders $i$ and types $t_i, t'_i$, guaranteeing that the interim form corresponds to a BIC mechanism.
\item $\sum_j \pi_{ij}(t_i) \cdot v_{ij}(t_i) - p_i (t_i) \geq 0$, for all bidders $i$ and types $t_i$, guaranteeing that the interim form corresponds to an interim IR mechanism.\footnote{Actually, this constraint is redundant as we will also enforce that the mechanism be ex-post IR to be considered feasible.}
\item $(O,\vec{\pi},\vec{p}) \in F([n+1]^m,\mathbb{R}^{m+1}_+,\textsc{Revenue})$, guaranteeing that the interim form corresponds to a feasible mechanism.
\end{enumerate}
\textbf{Maximizing:}
\begin{itemize}
\item $O$, the expected revenue.
\end{itemize}

The solution to this LP is the interim form of the optimal mechanism. The LP can be solved in polynomial time, so long as we have a poly-time separation oracle for the space $F([n+1]^m,\mathbb{R}^{m+1}_+,\textsc{Revenue})$. \cite{CaiDW13b} shows that this can be obtained via an algorithm for the related GOOP problem, which we instantiate in our setting below.

\paragraph{Budgeted-Additive Virtual Welfare Maximization} As discussed above, in order to find (approximately) optimal mechanisms for our setting, we need to study the purely algorithmic problem GOOP($[n+1]^m,\mathbb{R}^{m+1}_+,\textsc{Revenue}$), which we pose formally below.\\

\noindent\textbf{GOOP($[n+1]^m,\mathbb{R}^{m+1}_+,\textsc{Revenue}$)}:\\
\textsc{Input}: Values $v_{ij} \geq 0$ and virtual values $w_{ij} \in \mathbb{R}$ for all $i, j$. Budget $b_i\in \mathbb{R}_+$ and price multiplier $m_i\in \mathbb{R}$ for all $i$.\\
\textsc{Output}: An allocation $\vec{x} \in \{0,1\}^{mn}$ and prices $\vec{p}$ such that $\sum_i x_{ij} \leq 1$ for all $j$ (each item awarded at most once), $\sum_j x_{ij} v_{ij} \geq p_i$ (ex-post IR), $p_i \leq b_i$ (ex-post budget respecting), and $p_i \geq 0$ (no positive transfers).\\
\textsc{Goal}: Find $\arg \max_{\vec{x},\vec{p}}\{\sum_i (m_i+1)p_i  + \sum_{ij} x_{ij} w_{ij}\}$ (virtual revenue plus virtual welfare). \\
\\
Note that in the above formulation, we have folded cases where $\textsc{Revenue}$ evaluates to $-\infty$ into feasibility constraints on the output. We make two quick further observations about the structure of GOOP($[n+1]^m,\mathbb{R}^{m+1}_+,\textsc{Revenue}$), and call the reformulation \textbf{B}udgeted-\textbf{A}dditive \textbf{V}irtual \textbf{W}elfare \textbf{M}aximization (BAVWM). Also, for cleanliness, we will replace the input price multipliers $m_i$ by $m_i - 1$ so that the term in the objective will be $\sum_i m_i p_i$. This is w.l.o.g. as each $m_i$ could be any real number.

\begin{observation}\label{obs:1}
If $m_i > 0$, the optimal choice for $p_i$ is always $\min\{b_i, \sum_j x_{ij} v_{ij}\}$. If $m_i \leq 0$, the best choice for $p_i$ is $0$.
\end{observation}

\begin{observation}\label{obs:2}
For all possible solutions $(\vec{x}, \vec{p})$, the quality of $(\vec{x}, \vec{p})$ for the input instance $(\vec{v},\vec{w},\vec{b},\vec{m})$ is the same as for the instance $(\vec{v}',\vec{w},\vec{b},\vec{m})$ where $v'_{ij} = \min\{v_{ij}, b_i\}$, for all $i,j $.
\end{observation}

In light of these, we may set all negative $m_i$ to $0$, and all $v_{ij}$ to $\min\{v_{ij},b_i\}$ without changing the problem, leading to the following reformulation.\\
\\
\noindent\textbf{Budgeted-Additive Virtual Welfare Maximization}:\\
\textsc{Input}: Budget $b_i$ for all agents. Values $v_{ij} \in [0,b_i]$ for all agents and items. Price multiplier $m_i \geq 0$ for all agents, and virtual value $w_{ij} \in \mathbb{R}$ for all agents and items.\\
\textsc{Output}: An allocation $\vec{x} \in \{0,1\}^{mn}$ such that $\sum_i x_{ij} \leq 1$ for all $j$ (each item awarded at most once).\\
\textsc{Goal}: Find $\arg \max_{\vec{x}}\{\sum_i (m_i \min\{b_i, \sum_j x_{ij} v_{ij}\} + \sum_{j} x_{ij} w_{ij})\}$.\\

Note that in the above formulation, we no longer need to optimize over the price vector, due to Observation~\ref{obs:1}. The problem can now be interpreted as just a welfare maximization problem, where bidder $i$'s valuation function is the sum of a budgeted-additive function (with non-negative item values) and an additive function (with possibly negative item values). Also, note that we can re-formulate the above problem to remove the multipliers $(m_i)_i$ from the input and the objective, by incorporating them in the $b_i$'s and the $v_{ij}$'s. We choose to leave them in so that it is more transparent how the inputs to BAVWM are related to the types reported by the bidders of the mechanism output by the~\citet{CaiDW13b} reduction.

\subsection{The Generalized Assignment Problem}\label{sec:gap}
Our main technical result will make use of a rounding algorithm for the \emph{Generalized Assignment Problem}. We give here a statement of the problem and a rounding theorem due to Shmoys and Tardos~\cite{ShmoysT93}.\\
\\
\noindent\textbf{Generalized Assignment Problem}:\\
\textsc{Input}: Processing times $p_{ij} \in \mathbb{R}_+$ and costs $c_{ij} \in \mathbb{R}$ for all machines $i$ and jobs $j$, capacities $T_i$ for all machines $i$.\footnote{Traditionally, some consider only costs $c_{ij} \in \mathbb{R}_+$, but the result we cite applies for negative costs as well.}\\
\textsc{Output}: An allocation $\vec{x} \in \{0,1\}^{mn}$ of jobs to machines such that $\sum_i x_{ij} = 1$ for all $j$ (each job is assigned) and $\sum_j x_{ij} p_{ij} \leq T_i$ (each machine processes at most its capacity).\\
\textsc{Goal}: Find $\arg \max_{\vec{x}}\{\sum_{i,j} x_{ij} c_{ij}\}$ (total cost).\footnote{Traditionally, it makes sense to minimize total cost. As costs are possibly negative, the use of max or min is irrelevant.}\\

Now, we provide an LP due to Shmoys and Tardos that outputs a fractional solution at least as good as OPT.\\

\noindent\textbf{Variables:}
\begin{itemize}
\item $x_{ij}$, for all machines $i$ and jobs $j$, denoting the fraction of job $j$ assigned to machine $i$.
\end{itemize}
\textbf{Constraints:}
\begin{enumerate}
\item $\sum_i x_{ij}= 1$, for all $j$, guaranteeing that every job is processed exactly once.
\item $\sum_j x_{ij} \leq T_i$, for all $i$, guaranteeing that no machine's capacity is violated.
\item $x_{ij} = 0$ if $p_{ij} > T_i$.
\end{enumerate}
\textbf{Maximizing:}
\begin{itemize}
\item $\sum_{i,j} x_{ij} c_{ij}$, the total cost.
\end{itemize}

\begin{theorem}\label{thm:ST}(\cite{ShmoysT93}) The optimal fractional solution to the above LP can be rounded in polynomial time to an integral solution such that:
\begin{enumerate}
\item $\sum_i x_{ij} = 1$, for all $j$.
\item $\sum_j x_{ij} \leq 2T_i$, for all $i$.
\item $\sum_j x_{ij} c_{ij} \geq \text{OPT}$.
\end{enumerate}
\end{theorem}
\section{Main Results}
\label{sec:main}
In Section~\ref{sec:computation} below, we provide our main computational result: a poly-time approximation algorithm for BAVWM, which implies a poly-time truthful mechanism for revenue maximization that respects ex-post IR and ex-post budget constraints. In Section~\ref{sec:structure}, we detail the structure of the optimal mechanism in this setting, as well as our computationally efficient mechanism from Section~\ref{sec:computation}.
\subsection{Computational Results}\label{sec:computation}In this section, we provide a poly-time 3-approximation for BAVWM. We begin by writing a LP relaxation, allowing the designer to award fractions of items as long as the total fraction awarded doesn't exceed $1$. We split the fraction of item $j$ awarded to bidder $i$ into two parts, $\bar{x}_{ij}$ and $\hat{x}_{ij}$. Let $\bar{x}_{ij}$ denote the fraction of item $j$ assigned to agent $i$ before exceeding $b_i$. And let $\hat{x}_{ij}$ denote the fraction of item $j$ assigned after. In other words, if $x_{ij}$ is the fraction of item $j$ assigned to agent $i$, we have $\bar{x}_{ij} + \hat{x}_{ij} = x_{ij}$, $\sum_j \bar{x}_{ij} v_{ij} \leq b_i$, and $\sum_j \bar{x}_{ij} v_{ij} = b_i$ if for any $j$, $\hat{x}_{ij} > 0$. The idea is that assigning more of item $j$ to agent $i$ before exceeding his budget increases both terms in the ``goal'' above, but assigning more after exceeding the budget only affects the second term. The LP relaxation is as follows:\\

\noindent\textbf{Variables:}
\begin{itemize}
\item $\bar{x}_{ij}$, for all agents $i$ and items $j$, denoting the fraction of item $j$ assigned to agent $i$ contributing to both the budgeted-additive and additive terms in bidder $i$'s (virtual) welfare.
\item $\hat{x}_{ij}$, for all agents $i$ and items $j$, denoting the fraction of item $j$ assigned to agent $i$ contributing to just the additive term in bidder $i$'s (virtual) welfare.
\end{itemize}
\textbf{Constraints:}
\begin{enumerate}
\item $\sum_i (\bar{x}_{ij} + \hat{x}_{ij}) \leq 1$, for all $j$, guaranteeing that no item is allocated more than once.
\item $\sum_j \bar{x}_{ij}v_{ij} \leq b_i$, for all $i$, guaranteeing that contributions to the budgeted-additive term are not overcounted.
\end{enumerate}
\textbf{Maximizing:}
\begin{itemize}
\item $\sum_{ij} m_i \bar{x}_{ij}v_{ij} + \sum_{ij} w_{ij} (\bar{x}_{ij} + \hat{x}_{ij})$, the virtual welfare. Note that as each $m_i \geq 0$ and $v_{ij} \geq 0$, the optimal solution will never have $\hat{x}_{ij} > 0$ unless $\sum_j \bar{x}_{ij} = b_i$.
\end{itemize}

It is clear that any solution to BAVWM has a corresponding fractional solution to this LP. So the goal is to solve this LP and round the fractional solution to an integral one without too much loss. The idea is that the feasible region now looks pretty similar to that of the generalized assignment problem, asking for an assignment of jobs to machines such that the capacity of machine $i$ is at most $b_i$. We first prove the following rounding theorem, which is a near-direct application of Theorem~\ref{thm:ST}.

\begin{theorem}\label{thm:rounding}
The optimal fractional solution to the above LP can be rounded in polynomial time to an integral assignment such that:
\begin{enumerate}
\item $\sum_i (\bar{x}_{ij} + \hat{x}_{ij}) \leq 1$ for all $j$.
\item $\sum_j \bar{x}_{ij} v_{ij} \leq 2b_i$ for all $i$.
\item $\sum_{ij} m_i \bar{x}_{ij}v_{ij} + \sum_{ij} w_{ij} (\bar{x}_{ij} + \hat{x}_{ij}) \geq OPT$, where $OPT$ is the value of the LP.
\end{enumerate}
\end{theorem}

\begin{proof}
We show how to interpret our LP as an instantiation of a fractional LP for the generalized assignment problem, and then directly apply Theorem~\ref{thm:ST}. We use $p_{ij}$ to denote processing times, $c_{ij}$ to denote costs, and $T_i$ to denote capacities in the created generalized assignment problem instance.

\begin{itemize}
\item Machines:
\begin{enumerate}
\item A dummy machine, $0$.
\item For all bidders $i$, a hat machine $\hat{i}$ (corresponding to the hat variables in our LP).
\item For all bidders $i$, a bar machine $\bar{i}$ (corresponding to the bar variables in our LP).
\end{enumerate}
\item Jobs: A job $j$ for all items $j$.
\item Processing times and costs:
\begin{enumerate}
\item $p_{0j} = c_{0j} = 0$ for all $j$. $T_0 = 0$.
\item $\hat{p}_{ij} = 0$ for all $j$. $\hat{c}_{ij} = w_{ij}$ for all $j$. $\hat{T}_i = 0$.
\item $\bar{p}_{ij} = v_{ij}$. $\bar{c}_{ij} = m_i v_{ij} + w_{ij}$. $\bar{T}_i = b_i$. 
\end{enumerate}
\end{itemize}

The fractional LP referenced in Theorem~\ref{thm:ST} on this instance would then be (note that the capacity constraints for machines $0$ and all $\hat{i}$ are vacuously satisfied, and that there do not exist any $i, j$ for which $p_{ij} > T_i$ by Observation~\ref{obs:2}):\\

\noindent\textbf{Variables:}
\begin{itemize}
\item $x_{0j}$, for all jobs $j$, denoting the fraction of job $j$ assigned to machine $0$.
\item $\bar{x}_{ij}$, for all machines $i$ and jobs $j$, denoting the fraction of job $j$ assigned to machine $\bar{i}$.
\item $\hat{x}_{ij}$, for all machines $i$ and jobs $j$, denoting the fraction of job $j$ assigned to machine $\hat{i}$.
\end{itemize}
\textbf{Constraints:}
\begin{enumerate}
\item $x_{0j} + \sum_i (\bar{x}_{ij} + \hat{x}_{ij}) = 1$, for all $j$, guaranteeing that every job is allocated exactly once.
\item $\sum_j \bar{x}_{ij}v_{ij} \leq b_i$, for all $i$, guaranteeing that the total processing time on machine $\bar{i}$ is at most $b_i$.
\end{enumerate}
\textbf{Maximizing:}
\begin{itemize}
\item $\sum_{ij} m_i \bar{x}_{ij}v_{ij} + \sum_{ij} w_{ij} (\bar{x}_{ij} + \hat{x}_{ij})$, the cost.
\end{itemize}

It's clear that this LP is exactly the same as our LP, just with an additional dummy bidder $0$ who collects all unallocated fractions of items. By Theorem~\ref{thm:ST}, the optimal fractional solution to this LP can be rounded in polynomial time to an integral solution whose total cost is at least as large, but where the capacity of machine $\bar{i}$ could be as large as $2b_i$, which is exactly an integral allocation of items to bidders with the desired properties.
\end{proof}

After applying Theorem~\ref{thm:rounding}, we now have an integral solution that is at least as good as the optimum, except our solution is infeasible. It's infeasible because it's ``getting credit'' for {(virtual) welfare in the budgeted-additive term} that is perhaps up to twice the budget (i.e. up to $2b_i$). An ``obvious'' fix to this problem might be to take this integral solution and {only take credit for budgeted-additive values up to $b_i$}, thereby making the solution feasible again. Unfortunately, because the objective is mixed sign, the resulting solution doesn't provide any approximation guarantee.\footnote{\mattnote{Consider, for example, the following instance: there is one buyer and two items. $v_{11} = v_{12} = 3$, $b_1 = 3$, $w_{11} = w_{12} = -2$. Then the allocation that awards both items and ``gets credit'' for up to $2b_i$ is believed to have virtual welfare $2$. However, the correctly computed virtual welfare of this allocation is actually $-1$, which clearly provides no meaningful approximation. Instead we must develop a procedure that, on this instance, would allocate just one of the items}.} Instead, we provide a simple procedure to select a feasible suballocation of this infeasible one that loses a factor of 3. 

\begin{theorem}\label{thm:2}
Given an integral allocation $\vec{x}$ satisfying $\sum_i \bar{x}_{ij} + \hat{x}_{ij} \leq 1$ for all $j$, $\sum_j \bar{x}_{ij} v_{ij} \leq 2b_i$ for all $i$, and $\sum_{ij} m_i \bar{x}_{ij}v_{ij} + \sum_{ij} w_{ij} (\bar{x}_{ij} + \hat{x}_{ij}) = C$, one can find in poly-time an integral allocation $\vec{y}$ such that:
\begin{enumerate}
\item $\sum_i (\bar{y}_{ij} + \hat{y}_{ij}) \leq 1$ for all $j$.
\item $\sum_j \bar{y}_{ij} v_{ij} \leq b_i$ for all $i$.
\item $\sum_{ij} m_i \bar{y}_{ij}v_{ij} + \sum_{ij} w_{ij} (\bar{y}_{ij} + \hat{y}_{ij}) \geq C/3$.
\end{enumerate}
\end{theorem}
\begin{proof}
For each $i$, we wish to partition the set of items assigned to $i$ via $\bar{x}_{ij}$ (of the infeasible integral solution), $S$, into three disjoint sets $S_i^1,S_i^2,S_i^3$ such that $\sum_{j \in S_i^k} v_{ij} \leq b_i$ for all $k$. This is always possible: consider sorting the elements in decreasing order of $v_{ij}$ and greedily adding them one at a time to the $S_i^k$ with minimal weight so far. Assume for contradiction that some item $j^*$, when added, pushes some $S_i^k$ from below $b_i$ to above $b_i$. Then without $j^*$, each of $S_i^1,S_i^2,S_i^3$ must have had weight strictly larger than $b_i - v_{ij^*}$. As the total weight in all three (without $j^*$) is at most $2b_i - v_{ij^*}$, this means that $2b_i - v_{ij^*} > 3(b_i - v_{ij^*}) \Rightarrow v_{ij^*} > b_i/2$. But as we processed elements in decreasing order of $v_{ij}$, this would imply that $j^*$ was the third (or earlier) item processed, meaning that some set must have been empty, and $j^*$ couldn't have possibly pushed it over the limit (as $v_{ij} \leq b_i$ for all $j$). Therefore, at termination we must have $\sum_{j \in S_i^k} v_{ij} \leq b_i$ for all $k$. Now, define $k^* = \arg \max_{k}\{\sum_{j \in S_i^k}m_iv_{ij} + w_{ij}\}$. Let $\bar{y}_{ij} = 1$ iff $j \in S_i^{k^*}$, and $\hat{y}_{ij} = \hat{x}_{ij}$ for all $j$.

%any partition of the jobs into at least four sets, $S_i^1,\ldots,S_i^4,\ldots$ sorted so that $\sum_{j \in S_i^k} v_{ij} \geq \sum_{j \in S_i^{k+1}} v_{ij}$, and also such that $\sum_{j \in S_i^k} v_{ij} \leq b_i$. As $\sum_{j \in S} v_{ij} \leq 2b_i$, we clearly must have $\sum_{j \in S_i^3 \cup S_i^4} v_{ij} \leq b_i$, or else we would also have $\sum_{j \in S_i^1\cup S_i^2} v_{ij} > b_i$, forming a contradiction. So we can combine $S_i^3$ and $S_i^4$ into a single set and re-sort, while maintaining the invariant that $\sum_{j \in S_i^1} v_{ij} \leq b_i$. Starting with the partition that puts all items in their own set and iterating this process therefore finds a desired partition of at most three sets.  

It's clear that $\sum_j \bar{y}_{ij}v_{ij} \leq b_i$ for all $i$. As $\bar{y}_{ij} \leq \bar{x}_{ij}$ for all $i,j$, it's also clear that $\sum_i \bar{y}_{ij} + \hat{y}_{ij} \leq 1$ for all $j$. Finally, by choice of $k^*$ it's also clear that $\sum_{ij} (m_iv_{ij}+w_{ij})\bar{y}_{ij} \geq \sum_{ij} (m_i v_{ij} + w_{ij})\bar{x}_{ij}/3$, and therefore $\sum_{ij} m_i \bar{y}_{ij}v_{ij} + \sum_{ij} w_{ij} (\bar{y}_{ij} + \hat{y}_{ij}) \geq C/3$, as desired.
\end{proof}

Combining Theorems~\ref{thm:rounding} and~\ref{thm:2} yields a feasible, integral allocation that is a 3-approximation by rounding the fractional solution output by our LP, and it is easy to see that the entire procedure runs in polynomial time.

\begin{theorem}\label{thm:algorithm}
There is a poly-time 3-approximation algorithm for Budgeted-Additive Virtual Welfare Maximization, which is a reformulation of GOOP($[n+1]^m,\mathbb{R}^{m+1}_+,\textsc{Revenue}$). Therefore, for all $\epsilon > 0$, there is a poly-time $(\epsilon,3)$-approximation algorithm for BMeD($[n+1]^m,\mathbb{R}^{m+1}_+,\textsc{Revenue}$). Specifically, if $\ell$ is the input length to an instance of BMeD($[n+1]^m,\mathbb{R}^{m+1}_+,\textsc{Revenue}$), the algorithm terminates in time $\poly(\ell, 1/\epsilon)$ and succeeds with probability $1-\text{exp}(-\poly(\ell, 1/\epsilon))$. 
\end{theorem}

We conclude this section with a remark about the special case of a single (or small constant) number of items. Notice that BAVWM can be solved exactly by exhaustive search in time $\poly(n^m)$. If $m$ is a small constant, exhaustive search may be computationally feasible, resulting in an exact algorithm (instead of a $3$-approximation).

\begin{remark}
Budgeted-Additive Virtual Welfare Maximization can be solved exactly in time $\poly(n^m)$ by exhaustive search. Therefore, for all $\epsilon > 0$, there is an $(\epsilon,1)$-approximation algorithm for BMeD($[n+1]^m,\mathbb{R}^{m+1}_+,\textsc{Revenue}$). Specifically, if $\ell$ is the input length to an instance of  BMeD($[n+1]^m,\mathbb{R}^{m+1}_+,\textsc{Revenue}$), the algorithm terminates in time $\poly(\ell, n^m,1/\epsilon)$ and succeeds with probability $1-\text{exp}(-\poly(\ell,1/\epsilon))$. 
\end{remark}

Finally, we remark that the single-item case is \emph{especially} simpler than even the two item case. We refer the reader to~\cite{CaiDW12,CaiDW13b} for complete details, but essentially the sampling procedure that results in the $\epsilon$ error of Theorem~\ref{thm:CDW} can be replaced by an exact computation \emph{only} in the single item case (and not even in the two item case), and $\epsilon$ can be set to exactly $0$. 

\begin{remark}
	\label{remark:singleitemVWM} 
Budgeted-Additive Virtual Welfare Maximization with $m = 1$ can be solved exactly in time $\poly(n)$ by exhaustive search: there are only $n$ possible outcomes, corresponding to assigning the item to exactly one of the agents. Therefore, there is a $(0,1)$-approximation algorithm (i.e. an exact algorithm) for  BMeD($[n+1],\mathbb{R}^{2}_+,\textsc{Revenue}$) (i.e. the single item case). Specifically, if $\ell$ is the input length to an instance of  BMeD($[n+1],\mathbb{R}^{2}_+,\textsc{Revenue}$), the algorithm terminates in time $\poly(\ell)$, and succeeds with probability $1$.
\end{remark}

\subsection{Structural Results}\label{sec:structure}\mattnote{
In this section, we discuss the structure of the optimal mechanism, and of the computationally efficient mechanism from Section~\ref{sec:computation}. We begin by characterizing the optimal mechanism by combining Theorem~\ref{thm:CDWstructure} with Observation~\ref{obs:1}.}

\begin{theorem}\label{thm:structure}\mattnote{ In any BMeD($[n+1]^m,\mathbb{R}^{m+1}_+,\textsc{Revenue}$) instance, the optimal mechanism can be implemented as a distribution over virtual welfare maximizers. Specifically, there exists a distribution $\Delta$ over mappings $(f_1^\delta,\ldots,f_n^\delta)$. Each mapping $f_i^\delta$ maps types $(\vec{v}_i,b_i) \in \mathbb{R}_+^{m+1}$ to a multiplier $m_i^\delta(\vec{v}_i,b_i) \in \mathbb{R}_+$ and a vector $\vec{w}^\delta(\vec{v},b_i) \in \mathbb{R}^m$. Define $\phi_i^\delta$ to be the mapping that takes as input types $(\vec{v}_i,b_i) \in \mathbb{R}_+^{m+1}$ and outputs a valuation function $\phi_i^\delta(\vec{v}_i,b_i)(\cdot)$ with $\phi_i^\delta(\vec{v}_i,b_i)(S) = m_i^\delta(\vec{v}_i,b_i) \cdot \min\{b_i,\sum_{j \in S} v_{ij}\} + \sum_{j \in S} w^\delta_{ij}(\vec{v}_i,b_i)$. The allocation rule of the optimal mechanism first samples $(f_1^\delta,\ldots,f_n^\delta)$ from $\Delta$, and on profile $(\vec{v},\vec{b})$, allocates the items according to $\arg \max_{S_1\sqcup \ldots \sqcup S_n \subseteq [m]} \{\sum_i \phi_i^\delta(\vec{v}_i,b_i)(S_i)\}$. Furthermore, if $m_i^\delta(\vec{v}_i,b_i) > 0$, bidder $i$ is charged $\min\{b_i,\sum_{j \in S_i} v_{ij}\}$. If $m_i^\delta(\vec{v}_i,b_i) = 0$, then bidder $i$ is charged $0$.}
\end{theorem}

\begin{proof}
\mattnote{The proof starts with an application of Theorem~\ref{thm:CDWstructure} to the problem BMeD($[n+1]^m,\mathbb{R}^{m+1}_+,\textsc{Revenue}$). By Observation~\ref{obs:1}, the joint optimization over allocations $x$ and price vectors $\vec{p}$ can be accomplished by transforming the optimization into one that depends only on the allocation. Once the allocation is found, optimization of the price vector follows as in Observation~\ref{obs:1}. }
\end{proof}
\mattnote{
We remark that the virtual types involved in Theorem~\ref{thm:structure} have valuation functions that are the sum of a budgeted-additive function, and an additive function (the latter may have negative item values). We also note that the budgeted-additive component depends in a very structured way on the input type $(\vec{v}_i,b_i)$. Specifically, $b_i$ is turned into a hard cap on the bidder's maximum valuation instead of a hard budget on her ability to pay, and the additive valuation $\vec{v}_i$ is kept the same, forming a budgeted-additive function that is scaled by a positive multiplier $m_i$. The multiplier $m_i$ and additional values $\vec{w}_i$ may show little structure with respect to the input types (or perhaps none at all). } 

We also remark that the structure is especially simple in the case of a single item, because a budgeted-additive function for a single item is just a typical valuation function (where the bidder's value for the item is the minimum of her value and her budget). Specifically, the virtual type parameterized by $m_i^\delta(v_i,b_i)$ and $w_i(v_i,b_i)$ values the item at $m_i \min\{v_i,b_i\} + w_i(v_i,b_i)$. This observation leads to the following simplification:

\begin{remark}\label{remark:structure}
\mattnote{ In any BMeD($[n+1],\mathbb{R}^{2}_+,\textsc{Revenue}$) instance (i.e. the single item case), the optimal mechanism can be implemented as a distribution over virtual value maximizers. Specifically, there exists a distribution $\Delta$ over mappings $(f_1^\delta,\ldots,f_n^\delta)$. Each mapping $f_i^\delta$ maps types $(v_i,b_i) \in \mathbb{R}_+^{2}$ to an indicator bit $m_i^\delta(v_i,b_i) \in \{0,1\}$ and a virtual value $\phi_i^\delta(v_i,b_i)$. The allocation rule of the optimal mechanism first samples $(f_1^\delta,\ldots,f_n^\delta)$ from $\Delta$, and on profile $(\vec{v},\vec{b})$, allocates the item to any bidder $i^* \in \arg\max_{i} \{\phi_i^\delta(v_i,b_i)\}$ if her virtual value is non-negative, and doesn't allocate the item otherwise. Furthermore, if $m_{i^*}^\delta({v}_{i^*},b_{i^*}) =1$, bidder $i^*$ is charged $\min\{b_{i^*},v_{i^*}\}$. If $m_{i^*}^\delta({v}_{i^*},b_{i^*}) = 0$, then bidder $i^*$ is charged $0$.}
\end{remark}

\mattnote{We conclude with a statement regarding the format of our computationally efficient mechanisms from Section~\ref{sec:computation}. This is an instantiation of Algorithm~2 in~\cite{CaiDW13b}, which is used to prove Theorem~\ref{thm:CDW}. }

\begin{theorem}\label{thm:format} The mechanism providing the guarantee of Theorem~\ref{thm:algorithm} has the following format:
\\
\textbf{Phase One, Find the Mechanism:}
\begin{enumerate}
\item Write a linear program that optimizes revenue over the space of truthful, feasible interim forms (Section~\ref{sec:instantiation}). 
\item Pick an $\epsilon > 0$. Using the algorithm developed in Section~\ref{sec:computation}, and the reduction of~\cite{CaiDW13b}, solve this linear program approximately.
\item This yields an interim form corresponding to a mechanism that is an $(\epsilon,3)$-approximation.
\item The linear program also outputs auxiliary information in the form of a distribution $\Delta$ over mappings $(f_1^\delta,\ldots,f_n^\delta)$ of the same format from Theorem~\ref{thm:structure}.
\end{enumerate}

\textbf{Phase Two, Run the Mechanism:}
\begin{enumerate}
\item Sample a mapping from $\Delta$ (provided in Phase One).
\item On profile $(\vec{v},b)$, run the approximation algorithm of Section~\ref{sec:computation} for Budgeted-Additive Virtual Welfare Maximization, with input budgets $b_i$, input values $v_{ij}$, input price multipliers $m_i^\delta(\vec{v}_i,b_i)$, and input virtual values $w_{ij}^\delta(\vec{v}_i,b_i)$. Select this allocation.
\item If $m_i^\delta(\vec{v}_i,b_i) > 0$, charge bidder $i$ the minimum of their budget and their value for the items they receive. Otherwise, charge them nothing.
\end{enumerate}
\end{theorem}

Note that this mechanism has basically the same structure as the optimal mechanism, except that on every profile it only approximately maximizes virtual welfare (and we also first have to find the mechanism, which is completely described by the distribution $\Delta$). In the special case of a single item, the structure can again be simplified.

\begin{remark}
In the special case of a single item, the following algorithm finds the optimal mechanism in polynomial time:
\\
\textbf{Phase One, Find the Mechanism:}
\begin{enumerate}
\item Write a linear program that optimizes revenue over the space of truthful, feasible interim forms (Section~\ref{sec:instantiation}). 
\item Using the reduction of~\cite{CaiDW13b} and the observation in Remark \ref{remark:singleitemVWM} that Budgeted-Additive Virtual Welfare Maximization with $m = 1$ can be solved exactly, solve this linear program exactly.
This yields an interim form corresponding to the optimal mechanism. 
%that is an $(0,1)$-approximation (optimal).
\item The linear program also outputs auxiliary information in the form of a distribution $\Delta$ over mappings $(f_1^\delta,\ldots,f_n^\delta)$ of the same format from Remark~\ref{remark:structure}.
\end{enumerate}

\textbf{Phase Two, Run the Mechanism:}
\begin{enumerate}
\item Sample a mapping from $\Delta$ (provided in Phase One).
\item On profile $(\vec{v},b)$, award item $j$ to any bidder $i^* \in \arg\max_{i} \{\phi_i^\delta(v_i,b_i)\}$ if her virtual value is non-negative. Don't allocate item $j$ otherwise.
\item If $m_i^\delta(\vec{v}_i,b_i) =1$, charge bidder $i$ the minimum of their budget and their value for the items they receive. Otherwise, charge them nothing.
\end{enumerate}
\end{remark}

\bibliographystyle{acmsmall}
\bibliography{budgets}

\begin{thebibliography}{}

\bibitem[\protect\citeauthoryear{Abrams}{Abrams}{2006}]{Abrams2006}
{\sc Abrams, Z.} 2006.
\newblock Revenue maximization when bidders have budgets.
\newblock In {\em Proceedings of the Seventeenth Annual ACM-SIAM Symposium on
  Discrete Algorithm}. SODA '06. 1074--1082.

\bibitem[\protect\citeauthoryear{Benoit and Krishna}{Benoit and
  Krishna}{2001}]{BenoitKrishna}
{\sc Benoit, J.-p.} {\sc and} {\sc Krishna, V.} 2001.
\newblock Multiple-object auctions with budget constrained bidders.
\newblock {\em Rev. Econ. Stud\/}, 155--179.

\bibitem[\protect\citeauthoryear{Bhalgat, Gollapudi, and Munagala}{Bhalgat
  et~al\mbox{.}}{2013}]{Bhalgatetal}
{\sc Bhalgat, A.}, {\sc Gollapudi, S.}, {\sc and} {\sc Munagala, K.} 2013.
\newblock Optimal auctions via the multiplicative weight method.
\newblock In {\em Proceedings of the Fourteenth ACM Conference on Electronic
  Commerce}. EC '13. ACM, New York, NY, USA, 73--90.

\bibitem[\protect\citeauthoryear{Bhattacharya, Goel, Gollapudi, and
  Munagala}{Bhattacharya et~al\mbox{.}}{2010}]{BhattacharyaGGM10}
{\sc Bhattacharya, S.}, {\sc Goel, G.}, {\sc Gollapudi, S.}, {\sc and} {\sc
  Munagala, K.} 2010.
\newblock Budget constrained auctions with heterogeneous items.
\newblock In {\em STOC}. 379--388.

\bibitem[\protect\citeauthoryear{Borgs, Chayes, Immorlica, Mahdian, and
  Saberi}{Borgs et~al\mbox{.}}{2005}]{Borgsetal2005}
{\sc Borgs, C.}, {\sc Chayes, J.~T.}, {\sc Immorlica, N.}, {\sc Mahdian, M.},
  {\sc and} {\sc Saberi, A.} 2005.
\newblock Multi-unit auctions with budget-constrained bidders.
\newblock In {\em ACM Conference on Electronic Commerce}. 44--51.

\bibitem[\protect\citeauthoryear{Cai, Daskalakis, and Weinberg}{Cai
  et~al\mbox{.}}{2012}]{CaiDW12}
{\sc Cai, Y.}, {\sc Daskalakis, C.}, {\sc and} {\sc Weinberg, S.~M.} 2012.
\newblock An algorithmic characterization of multi-dimensional mechanisms.
\newblock In {\em the 44th Symposium on Theory of Computing Conference (STOC)}.

\bibitem[\protect\citeauthoryear{Cai, Daskalakis, and Weinberg}{Cai
  et~al\mbox{.}}{2013}]{CaiDW13b}
{\sc Cai, Y.}, {\sc Daskalakis, C.}, {\sc and} {\sc Weinberg, S.~M.} 2013.
\newblock Understanding incentives: Mechanism design becomes algorithm design.
\newblock In {\em Proceedings of the 2013 IEEE 54th Annual Symposium on
  Foundations of Computer Science}. FOCS '13. 618--627.

\bibitem[\protect\citeauthoryear{Chakrabarty and Goel}{Chakrabarty and
  Goel}{2010}]{ChakrabartyGoel}
{\sc Chakrabarty, D.} {\sc and} {\sc Goel, G.} 2010.
\newblock On the approximability of budgeted allocations and improved lower
  bounds for submodular welfare maximization and {GAP}.
\newblock {\em {SIAM} J. Comput.\/}~{\em 39,\/}~6, 2189--2211.

\bibitem[\protect\citeauthoryear{Chawla, Malec, and Malekian}{Chawla
  et~al\mbox{.}}{2011}]{ChawlaMM11}
{\sc Chawla, S.}, {\sc Malec, D.~L.}, {\sc and} {\sc Malekian, A.} 2011.
\newblock Bayesian mechanism design for budget-constrained agents.
\newblock In {\em ACM Conference on Electronic Commerce}. 253--262.

\bibitem[\protect\citeauthoryear{Che and Gale}{Che and
  Gale}{1998}]{CheGale1998}
{\sc Che, Y.-K.} {\sc and} {\sc Gale, I.} 1998.
\newblock Standard auctions with financially constrained bidders.
\newblock {\em Review of Economic Studies\/}~{\em 65,\/}~1, 1--21.

\bibitem[\protect\citeauthoryear{Che and Gale}{Che and
  Gale}{2000}]{CheGale2000}
{\sc Che, Y.-K.} {\sc and} {\sc Gale, I.} 2000.
\newblock The optimal mechanism for selling to a budget-constrained buyer.
\newblock {\em Journal of Economic Theory\/}~{\em 92,\/}~2, 198--233.

\bibitem[\protect\citeauthoryear{Daskalakis and Weinberg}{Daskalakis and
  Weinberg}{2015}]{DaskalakisW15}
{\sc Daskalakis, C.} {\sc and} {\sc Weinberg, S.~M.} 2015.
\newblock Bayesian truthful mechanisms for job scheduling from bi-criterion
  approximation algorithms.
\newblock In {\em Proceedings of the Twenty-Sixth Annual ACM-SIAM Symposium on
  Discrete Algorithms}. SODA '15. 1934--1952.

\bibitem[\protect\citeauthoryear{Devanur, Ha, and Hartline}{Devanur
  et~al\mbox{.}}{2013}]{DevanurHH13}
{\sc Devanur, N.~R.}, {\sc Ha, B.~Q.}, {\sc and} {\sc Hartline, J.~D.} 2013.
\newblock Prior-free auctions for budgeted agents.
\newblock In {\em Proceedings of the Fourteenth ACM Conference on Electronic
  Commerce}. EC '13. 287--304.

\bibitem[\protect\citeauthoryear{Dobzinski, Lavi, and Nisan}{Dobzinski
  et~al\mbox{.}}{2008}]{DobzinskiLaviNisan}
{\sc Dobzinski, S.}, {\sc Lavi, R.}, {\sc and} {\sc Nisan, N.} 2008.
\newblock Multi-unit auctions with budget limits.
\newblock In {\em Proceedings of the 2008 49th Annual IEEE Symposium on
  Foundations of Computer Science}. FOCS '08. IEEE Computer Society,
  Washington, DC, USA, 260--269.

\bibitem[\protect\citeauthoryear{Goel, Mirrokni, and Paes~Leme}{Goel
  et~al\mbox{.}}{2012}]{GoelMirrokniPaes-Leme}
{\sc Goel, G.}, {\sc Mirrokni, V.}, {\sc and} {\sc Paes~Leme, R.} 2012.
\newblock Polyhedral clinching auctions and the adwords polytope.
\newblock In {\em Proceedings of the 44th symposium on Theory of Computing}.
  STOC '12. Association for Computing Machinery, 107--122.

\bibitem[\protect\citeauthoryear{Laffont and Robert}{Laffont and
  Robert}{1996}]{LaffontRobert}
{\sc Laffont, J.-J.} {\sc and} {\sc Robert, J.} 1996.
\newblock Optimal auction with financially constrained buyers.
\newblock {\em Economics Letters\/}~{\em 52,\/}~2, 181--186.

\bibitem[\protect\citeauthoryear{Malakhov and Vohra}{Malakhov and
  Vohra}{2005}]{MalakhovVohra}
{\sc Malakhov, A.} {\sc and} {\sc Vohra, R.~V.} 2005.
\newblock {Optimal Auctions for Asymmetrically Budget Constrained Bidders}.
\newblock Discussion Papers 1419, Northwestern University, Center for
  Mathematical Studies in Economics and Management Science. Dec.

\bibitem[\protect\citeauthoryear{Maskin}{Maskin}{2000}]{Maskin}
{\sc Maskin, E.} 2000.
\newblock Auctions, development, and privatization: Efficient auctions with
  liquidity-constrained buyers.
\newblock {\em European Economic Review\/}~{\em 44,\/}~4-6, 667--681.

\bibitem[\protect\citeauthoryear{Myerson}{Myerson}{1981}]{Myerson}
{\sc Myerson, R.} 1981.
\newblock Optimal auction design.
\newblock {\em Mathematics of Operations Research\/}~{\em 6,\/}~1, 58--73.

\bibitem[\protect\citeauthoryear{Pai and Vohra}{Pai and Vohra}{2014}]{PaiVohra}
{\sc Pai, M.~M.} {\sc and} {\sc Vohra, R.} 2014.
\newblock Optimal auctions with financially constrained buyers.
\newblock {\em Journal of Economic Theory\/}~{\em 150,\/}~C, 383--425.

\bibitem[\protect\citeauthoryear{Shmoys and Tardos}{Shmoys and
  Tardos}{1993}]{ShmoysT93}
{\sc Shmoys, D.~B.} {\sc and} {\sc Tardos, {\'E}.} 1993.
\newblock {Scheduling Unrelated Machines with Costs}.
\newblock In {\em the 4th Symposium on Discrete Algorithms (SODA)}.

\end{thebibliography}

\end{document}